\documentclass[journal]{IEEEtran}
\IEEEoverridecommandlockouts
\usepackage{cite}
\usepackage{longtable}
\usepackage{algorithmic}
\usepackage{graphicx}
\usepackage{caption}
\usepackage{textcomp}
\usepackage{enumerate}
\usepackage{latexsym}
\usepackage{amsmath}
\usepackage{amsthm}
\usepackage{amssymb}
\usepackage{epsfig}
\usepackage{overpic}
\usepackage{multirow}
\usepackage[table]{xcolor}
\usepackage{colortbl,array}
\usepackage{multirow,bigdelim}
\usepackage{hyperref}
\usepackage{subfigure}
\usepackage{booktabs}
\usepackage{fancyhdr}
\usepackage{nomencl}
\usepackage[flushleft]{threeparttable}
\usepackage{figsize}
\usepackage{mathtools}
\usepackage{algorithm}
\usepackage{soul}
\usepackage[utf8]{inputenc}
\usepackage[english]{babel}
\usepackage{arydshln}
\usepackage{tikz}
\usepackage{tikz-3dplot}
\usetikzlibrary{positioning, arrows.meta}
\newtheorem{thm}{Theorem}
\newtheorem{lem}{Lemma}
\newtheorem{rem}{Remark}
\newtheorem{cor}{Corollary}
\newtheorem{asm}{Assumption}
\newtheorem{defa}{Definition}
\newtheorem{ex}{Example}

\title{\LARGE \bf
Stability and Robustness of Time-Varying Opinion Dynamics: A Graph-Theoretic Approach
}

\author{M. Hossein Abedinzadeh$^{1}$ and Emrah Akyol$^{1}$%
\thanks{*This research is supported by the NSF via (CAREER)  CCF-2048042.}%
\thanks{$^{1}$ M. Hossein Abedinzadeh and Emrah Akyol are with the Department of Electrical and Computer Engineering, Binghamton University (SUNY), Binghamton, NY, 13902 USA. (e-mail: {\tt\small mabedin3@binghamton.edu}, {\tt\small eakyol@binghamton.edu})}%
}

\begin{document}
\maketitle
\thispagestyle{empty}
\pagestyle{empty}

\begin{abstract}
The time-varying Friedkin--Johnsen (TVFJ) model provides a principled foundation for analyzing realistic, evolving social networks by capturing both persistent individual biases and adaptive influences. Building on this model, we develop a graph-theoretic framework that certifies stability and characterizes long-term opinion dynamics. We introduce \emph{defected} and \emph{weakly defected} temporal graphs (DTGs/WDTGs) as topological certificates that translate temporal connectivity and stubborn influence into contraction bounds on the state-transition matrix. Using these concepts, we establish: (i) \emph{asymptotic stability} of the TVFJ model when DTGs recur infinitely often; (ii) \emph{exponential stability} of the semi-periodic Friedkin--Johnsen (SPFJ) model, with explicit growth and decay constants; and (iii) \emph{asymptotic stability} of a trust-based Friedkin--Johnsen (TBFJ) variant under the weaker condition of infinitely many WDTGs. We further prove the \emph{boundedness of the $\omega$-limit set}, ensuring that long-run opinions remain within the convex hull of innate beliefs. For periodically switching Friedkin--Johnsen (PSFJ) models with period \(p\), we derive a \(p\)-LTI decomposition that yields closed-form equilibria and establishes the tight bound \(|\omega|\leq p\). Finally, we show that \emph{robustness} holds: exponential stability in SPFJ and PSFJ systems persists under bounded perturbations of the interaction structure. Collectively, these results unify algebraic stability tests with interpretable, topology-driven analysis, providing scalable and resilient tools for reasoning about opinion formation in dynamic networks.
\end{abstract}


\section{Introduction}
With the rapid growth and increasing complexity of large-scale social networks, there is a pressing need for principled models that capture evolving interaction structures and provide scalable tools for analyzing opinion dynamics and collective behavior.  

Classical opinion--dynamics models laid the foundation for understanding how local interactions shape collective outcomes~\cite{french1956formal,degroot1974reaching,friedkin1999social,abelson1964mathematical,taylor1968towards}. Convergence has been extensively studied through two complementary approaches: algebraic methods and graph-theoretic arguments~\cite{degroot1974reaching,ren2005consensus,parsegov2016novel,demarzo2003persuasion,xia2011clustering}. Much of this literature assumes a \emph{time-invariant} interaction graph, justified when (i) the network topology remains effectively static over the horizon of interest, and (ii) agents are \emph{boundedly rational}, updating via heuristics with limited computation rather than global optimization~\cite{acemoglu2011opinion}.  

In contrast, real networks evolve: nodes and ties may appear, disappear, or change strength, making the interaction topology inherently time varying. Simultaneously, research on human--AI interaction has accelerated~\cite{askarisichani2022predictive,o2022human,wilson2018collaborative,kleinberg2018human,seeber2020machines,amershi2019guidelines}, while AI-mediated curation (e.g., recommender systems, assistant agents) has become pervasive on social platforms~\cite{wu2022graph,ko2022survey,10506571}. Humans are typically modeled as \emph{boundedly rational}, relying on local heuristics, whereas AI systems often employ near-global optimization rules, further challenging the validity of classical assumptions~\cite{chaney2018algorithmic,bakshy2015exposure}.  

These developments motivate a principled framework for time-varying opinion dynamics that: (i) remains consistent with validated classical models while incorporating temporal variability, (ii) captures both human and AI agents together with their distinct forms of bounded rationality, (iii) admits a transparent analytical foundation for rigorous stability and robustness analysis, and (iv) provides interpretable, data-integrated tools applicable across scales from controlled experiments to large online platforms.  

Time-varying extensions are typically studied under two paradigms: \emph{discrete-time snapshots} and \emph{continuously updated} networks~\cite{aggarwal2014evolutionary}. We focus on the discrete-time setting. Two natural candidates for temporal extension are the time-invariant French--DeGroot (TIFD) and time-invariant Friedkin--Johnsen (TIFJ) models. The TIFD model has no exogenous pull toward prejudices~\cite{french1956formal,degroot1974reaching}, whereas the TIFJ model incorporates persistent attraction to innate opinions via stubborn agents~\cite{friedkin1999social}. We study their time-varying counterparts, TVFD and TVFJ, which capture evolving interactions in social and human--AI networks.  

In TIFD, convergence typically leads to consensus or clustering, with outcomes determined by initial opinions~\cite{proskurnikov2017tutorial}. Extensions remain largely confined to consensus analysis~\cite{rainer2002opinion,olfati2006flocking,ren2005consensus,lorenz2005stabilization,blondel2005convergence}, often using the union graph for connectivity guarantees, an approach that overlooks \emph{temporal connectivity persistence} and requires restrictive assumptions such as universal self-loops.  

In contrast, in TIFJ, consensus is only a special case: if the system is asymptotically stable, the trajectories converge to a unique equilibrium regardless of the initial conditions. Extensions to state-dependent and time-varying settings have been studied~\cite{10313365,proskurnikov2017opinion,disaro2024balancing}. Unlike the time-invariant case, where stability coincides with convergence to a fixed point, the TVFJ model may exhibit convergence or \emph{containment}, depending on temporal connectivity~\cite{proskurnikov2017opinion}.  

In real systems, the assumption of non-stubborn agents rarely holds, limiting the explanatory power of consensus-based analysis. The algebraic stability results for TVFJ~\cite{proskurnikov2017opinion} provide general guarantees but often obscure the underlying mechanisms, especially in large-scale networks. This motivates graph-theoretic heuristics that certify stability in terms of temporal propagation of stubborn influence, offering interpretable and scalable insights absent from purely algebraic arguments.  

To this end, we model interactions as a temporal multiplex network~\cite{holme2012temporal,kivela2014multilayer}, where each layer corresponds to a time-stamped interaction graph. To our knowledge, multiplex structures have not yet been systematically used as \emph{graph-theoretic certificates of stability} for TVFJ dynamics.  

We analyze two complementary regimes. First, in the general TVFJ model, where both weights and susceptibilities evolve, we introduce \emph{defected temporal graphs} (DTGs) as topological witnesses of contraction, certifying stability by ensuring stubborn influence percolates across the network. Second, we extend to \emph{trust-based Friedkin--Johnsen models} (TBFJ), where influence is filtered through a static trust matrix and susceptibilities depend only on locally perceived neighborhoods. Here we define \emph{weakly defected temporal graphs} (WDTGs), which require only that each non-stubborn agent remains temporally connected to a stubborn agent, regardless of influence strength.  

Our results show that infinitely recurring DTGs (or WDTGs, in the TBFJ setting) guarantee asymptotic stability. This unifies algebraic contraction arguments with temporal connectivity, providing a transparent graph-theoretic foundation for analyzing time-varying opinion dynamics.  

Beyond stability, we study long-term behavior through the $\omega$-limit set. For periodically switching FJ (PSFJ) models with period \(p\), we derive an exact $p$-LTI decomposition, closed-form equilibria, and the tight bound $|\omega| \leq p$. For semi-periodic defected networks, we establish exponential stability with explicit growth and decay rates. In general, we prove that the $\omega$-limit set of TVFJ dynamics is bounded within the convex hull of innate opinions, guaranteeing strong containment of long-run trajectories.  

Finally, we address robustness. Since real networks exhibit noise, structural uncertainty, and temporal irregularities, we show that exponential stability in SPFJ and PSFJ settings persists under bounded perturbations. This bridges idealized models with practical, imperfect networks.  

\textbf{Related Work.} Research on time-varying networks has largely focused on consensus dynamics. For example, ~\cite{rainer2002opinion,olfati2006flocking,ren2005consensus,lorenz2005stabilization,blondel2005convergence} analyze convergence relying on the union graph to guarantee connectivity over aggregated intervals, but these approaches overlook the critical aspect of \emph{temporal connectivity persistence}, often requiring restrictive assumptions such as universal self-loops. More recently, multilayer and multiplex frameworks have been employed in the study of opinion dynamics~\cite{ju2025influence,ShiuNaghizadeh2025Allerton,hu2017opinion}. In particular,~\cite{ShiuNaghizadeh2025Allerton} examines multiplex structures with heterogeneous update frequencies across layers; however, these works use multilayer modeling primarily to represent multiple coexisting social networks rather than as a tool for investigating system stability and are all formulated under the TVFD model. In contrast, stability analysis of the TVFJ model has been studied in~\cite{proskurnikov2017opinion,disaro2024balancing}. The work in~\cite{disaro2024balancing} considers a special case in which the susceptibility matrix is fixed and the influence matrix takes values in a finite set, with entries updated according to a homophily principle; they prove that the opinion matrix converges asymptotically to a constant solution if and only if the influence matrix eventually becomes constant. In contrast,~\cite{proskurnikov2017opinion} analyzes the general TVFJ model without additional assumptions, providing sufficient algebraic conditions for stability. Our work extends~\cite{proskurnikov2017opinion} by introducing a transparent graph-theoretic framework that not only generalizes their algebraic condition but also yields interpretable topology-driven stability certificates.

\textbf{Contributions.} The main technical contributions of this work are:  
\begin{itemize}
    \item Development of a graph-theoretic framework for analyzing stability in time-varying opinion dynamics.  
    \item Derivation of sufficient conditions for asymptotic stability of the TVFJ model and exponential stability in semi-periodic defected networks.  
    \item Introduction of a trust-based opinion dynamics model whose stability is guaranteed under the weaker WDTG condition.  
    \item Characterization of the $\omega$-limit set of PSFJ models, including a $p$-LTI decomposition, closed-form equilibria, and the tight bound $|\omega| \leq p$.  
    \item Establishment of robustness guarantees showing that exponential stability persists under bounded perturbations of interaction weights.  
\end{itemize}

\textbf{Organization.} Section~\ref{sec:OpinionDynamicsModel} formalizes the TVFJ and TBFJ models and introduces the relevant stability concepts. Section~\ref{sec:TemporalNetworkStructure} defines temporal paths, DTGs, and WDTGs. Section~\ref{sec:omegaBoundedness} establishes the boundedness of the $\omega$-limit sets. Section~\ref{sec:TVFJStability} proves asymptotic stability of the TVFJ model and exponential stability of SPFJ models under infinitely many DTGs. Section~\ref{sec:TrustBasedStability} extends the analysis to TBFJ dynamics and PSFJ models. Section~\ref{sec: RobustnessofStabilityUnderPerturbations} develops robustness guarantees. Finally, Section~\ref{Conclution} concludes with directions for multi-dimensional, strategic, and behavioral extensions.

\section{Preliminaries and Notation}
Throughout this paper, we denote vectors by bold lowercase letters (e.g., \( \mathbf{x} \)) and matrices by uppercase letters (e.g., \( A \)). The \( (i,j) \)-th entry of a matrix \( A \in \mathbb{R}^{n \times m} \) is denoted by \( a_{ij} \). If \( A \) is a diagonal matrix, we use \( a_{i} \) to denote its \( i \)-th diagonal element. Vector and matrix norms are denoted by \( \|\cdot\| \), without specifying a particular type of norm. For clarity in analysis and technical proofs, we adopt the \( \ell_1 \) norm:
\[
\|\mathbf{x}\| := \sum_i |x_i|, \qquad \|A\| := \max_i \sum_j |a_{ij}|.
\]

A matrix $A \in [0,1]^{n \times m}$ is row-stochastic if $ \sum_{j=1}^m{{a_{ij}}}=1$ and sub-row-stochastic if $ \sum_{j=1}^m{{a_{ij}}}\leq 1$ for all $ i \leq n$.

\section{Opinion Dynamics Model}
\label{sec:OpinionDynamicsModel}

Consider a social network consisting of a finite set of agents \( \mathcal{V} = \{\mathrm{v}_1, \ldots, \mathrm{v}_n\} \). Each agent maintains an \emph{expressed opinion} \( x_i[t] \in [0,1] \) at time \( t \), and an \emph{innate opinion} \( s_i \in [0,1] \) that remains fixed. Collect these into the vectors
\begin{align*}
    \mathbf{x}[t] &= \begin{bmatrix} x_1[t] & \cdots & x_n[t] \end{bmatrix}^\top, \quad    \mathbf{s} &= \begin{bmatrix} s_1 & \cdots & s_n \end{bmatrix}^\top.
\end{align*}
The opinion dynamics process is described by the time-varying Friedkin--Johnsen (TVFJ) model~\cite{proskurnikov2017opinion}:
\begin{equation}
    \mathbf{x}[t+1] = \Lambda[t] W[t] \mathbf{x}[t] + (I - \Lambda[t])\mathbf{s},
    \label{eq:update_equation}
\end{equation}
where \( \Lambda[t] = \mathrm{diag}(\lambda_1[t], \ldots, \lambda_n[t]) \) is the diagonal susceptibility matrix with \( \lambda_i[t] \in [0,1] \), and \( W[t] \in \mathbb{R}^{n \times n} \) is row-stochastic. Each \( W[t] \) induces a time-varying adjacency matrix \( A[t] \in \{0,1\}^{n \times n} \) with
\[
a_{ij}[t] = \begin{cases}
1, & w_{ij}[t] > 0, \\ 0, & \text{otherwise},
\end{cases}
\]
and the corresponding neighbor set of agent \( \mathrm{v}_i \) at time \( t \),
\( \mathcal{N}_i[t] := \{ \mathrm{v}_j \in \mathcal{V} \mid a_{ij}[t] = 1 \} \).

The solution of \eqref{eq:update_equation}, initialized at \( t_0 \), is
\begin{equation}
    \mathbf{x}[t] = \Phi(t, t_0)\,\mathbf{x}[t_0] + \sum_{\tau = t_0}^{t-1} \Phi(t, \tau+1)\,(I - \Lambda[\tau])\,\mathbf{s},
    \label{eq:fj_solution}
\end{equation}
where the state transition matrix \(\Phi(t,\tau)\) is
\[
\Phi(t,\tau) :=
\begin{cases}
\prod_{k=\tau}^{t-1} \Lambda[k] W[k], & t > \tau, \\
I, & t = \tau.
\end{cases}
\]

For a trajectory \( \{\mathbf{x}[t]\}_{t=0}^\infty \) with initial state \( \mathbf{x}[0] = \mathbf{x}_0 \), the \emph{\(\omega\)-limit set} is
\[
\omega(\mathbf{x}_0) := \left\{ \mathbf{y} \in \mathbb{R}^n ~\middle|~ \exists\, \{t_k\}\subset\mathbb{N},\ \lim_{k \to \infty} \mathbf{x}[t_k] = \mathbf{y} \right\}.
\]
This set capturing the asymptotic effects of both the initial state and the evolving influence structure.

\begin{defa}[Trust-Based Opinion Dynamics]
Let \( A[t] \) be a time-varying adjacency matrix and \( \hat{W} \in \mathbb{R}^{n \times n} \) a fixed row-stochastic trust matrix. Define
\[
w_{ij}[t] := \frac{a_{ij}[t]\hat{w}_{ij}}{\sum_{k \in \mathcal{N}_i[t]} a_{ik}[t]\hat{w}_{ik}}, 
\quad 
\lambda_i[t] := f_i(\mathcal{N}_i[t]),
\]
where \( f_i : 2^{\mathcal{V}} \to [0,1] \) with \( f_i(\emptyset)=0 \).  
The dynamics retain~\eqref{eq:update_equation}, with \( W[t] \) and \( \Lambda[t] \) given above.
\label{defa:Trust_Based_TVFJ}
\end{defa}

\begin{rem}
    Definition \ref{defa:Trust_Based_TVFJ} captures two conditions: humans exhibit \emph{bounded rationality}, so both their influence weights \(W[t]\) and susceptibilities \(\Lambda[t]\) vary over time \cite{acemoglu2011opinion}; AI agents, by contrast, have global access but operate under \emph{rational inattention}, distributing fixed trust weights across available neighbors in order to increase utility while reducing processing cost \cite{mackowiak2023rational}.
\end{rem}

\begin{defa}[Asymptotic and Exponential Stability]
The opinion dynamics \eqref{eq:update_equation} is \emph{asymptotically stable (AS)} if, for all \( \tau \geq 0 \),
\[
\lim_{t \to \infty} \|\Phi(t,\tau)\| = 0.
\]
It is \emph{exponentially stable (ES)} if there exist constants \( c > 0 \) and \( \gamma \in (0,1) \) such that
\[
\|\Phi(t,\tau)\| \leq c\,\gamma^{t-\tau}, \quad \forall\, t \geq \tau \geq 0.
\]
Here, \( c \) is the \emph{growth factor}, and \( \gamma \) the \emph{decay rate}.
\end{defa}

\begin{defa}[Stubborn and Strictly Stubborn Agents]
Agent \( \mathrm{v}_i \) is \emph{stubborn} at time \( t \) if \( \lambda_i[t] < 1 \).  
It is \emph{strictly stubborn} (or \(\epsilon\)-stubborn) if there exists \( \epsilon > 0 \) such that \( \lambda_i[t] \leq 1-\epsilon \).
\label{defa:strictly_stubborn}
\end{defa}

\begin{asm}
The susceptibility of agent \( \mathrm{v}_i \) satisfies \( \lambda_i[t] = 0 \) if and only if its influence row is zero, i.e., \( w_{ij}[t] = 0 \) for all  \( \mathrm{v}_j \in \mathcal{V} \).
\label{asm:asm1}
\end{asm}

\begin{asm}
To avoid trivial stability, there is no time \( \tau \) such that the susceptibility matrix vanishes identically, i.e., \( \Lambda[\tau] = \mathbf{0} \).
\label{asm:asm2}
\end{asm}

\section{Temporal Network Structure}
\label{sec:TemporalNetworkStructure}

At time \( t \), interactions are described by the influence matrix \( W[t] \), which induces the directed graph \( \mathcal{G}[t] = (\mathcal{V}, \mathcal{E}[t]) \) with edge set \( \mathcal{E}[t] := \{ (\mathrm{v}_j,\mathrm{v}_i) \mid w_{ij}[t] > 0 \} \). A \emph{temporal graph} \( \mathcal{G}_{t_1}^{t_2} \) corresponding to the update equation~\eqref{eq:update_equation} in the time interval \([t_1, t_2]\) is defined as a multiplex-directed graph \(\mathcal{G}_{t_1}^{t_2} = \{ \mathcal{G}[k] \}_{k=t_1}^{t_2}\), where each layer \( \mathcal{G}[k] \) represents the topology of the network at the time step \( k \).  A \emph{multiplex network} is a multilayer network with a fixed node set and evolving connectivity across layers ~\cite{kivela2014multilayer}. A triplet \( (\mathrm{v}_j, \mathrm{v}_i, t) \) is a \emph{temporal edge} if \( (\mathrm{v}_j, \mathrm{v}_i) \in \mathcal{E}[t] \), and is said to be \emph{\(w\)-edge} if \( w_{ij}[t] \geq w \) for some threshold \( w > 0 \). A \emph{temporal path} from node \( \mathrm{v}_j \) to node \( \mathrm{v}_i \) over the interval \( [t_{k_0}, t_{k_m}] \) is a sequence of temporal edges connecting intermediate nodes across consecutive time steps. Formally, it is represented as:
\begin{align*}
(\mathrm{v}_j , \mathrm{v}_i, t_{k_0}, t_{k_m}) = \bigl\{&(\mathrm{v}_{j}, \mathrm{v}_{i_1}, t_{k_0}),(\mathrm{v}_{i_1}, \mathrm{v}_{i_2}, t_{k_1}), \dots\\
&(\mathrm{v}_{i_{m}}, \mathrm{v}_i, t_{k_m}) \bigr\}.
\end{align*}
Let \( (\mathrm{v}_j , \mathrm{v}_i, t_1, t_2) \) denote a temporal path. Such a path is called an \emph{s-path} if the node \( \mathrm{v}_j \) is stubborn. It is called an \emph{influential-path} if \( \mathrm{v}_j \) is \(\epsilon\)-stubborn and all edges along the path are \(w\)-edges.

\begin{defa}[Defected and Weakly Defected Temporal Graphs]
A temporal graph \( \mathcal{G}_{t_0}^{t_d} \) is a \emph{weakly defected temporal graph (WDTG)} if there exists \( k \in [t_0,t_d) \) such that, in layer \( \mathcal{G}[k] \), every agent \( \mathrm{v}_i \) is either stubborn or connected to a stubborn agent via a finite \( s \)-path fully contained in \( [t_0,t_d) \).  

It is a \emph{defected temporal graph (DTG)} if there exists \( k \in [t_0,t_d) \) such that, in layer \( \mathcal{G}[k] \), every agent \( \mathrm{v}_i \) is either \(\epsilon\)-stubborn or connected to an \(\epsilon\)-stubborn agent via a finite influential-path entirely within \( [t_0,t_d) \).
\label{defa:dtg_wdtg}
\end{defa}

\begin{figure}[H]
\centering
\includegraphics[width=0.25\textwidth]{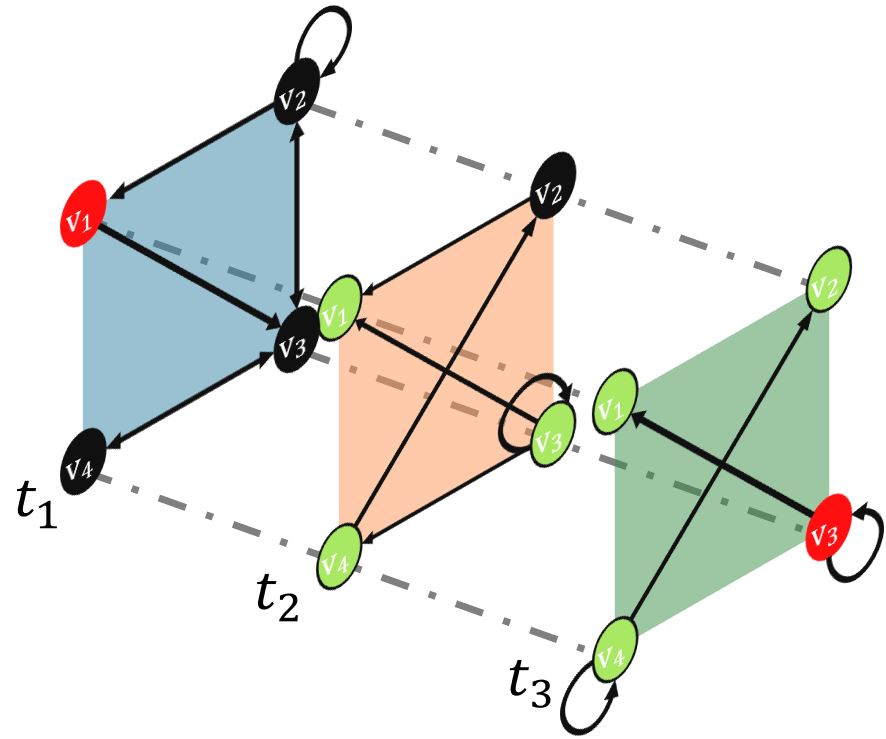}
\caption{Layer \( \mathcal{G}(t_3) \) is a defected layer. All edges represent influential connections; red nodes indicate strictly stubborn agents, and green nodes are those temporally connected to the strictly stubborn agent via influential paths.}
\label{fig:defected-layer}
\end{figure}

Figure~\ref{fig:defected-layer} illustrates a sequence of time-indexed interaction layers that represent a temporal multiplex network. In this example, layer \( \mathcal{G}(t_3) \) contains a strictly stubborn agent (node \( \mathrm{v}_3 \), shown in red), and all other agents are temporally connected to the strictly stubborn agent \(\mathrm{v}_1\) in \( \mathcal{G}(t_1) \) through influential-path. This demonstrates the idea of a defected-layer, making this a DTG.

\begin{rem}
A defected temporal graph (DTG) (or, more generally, a WDTG) serves as a structural certificate: within one layer of the temporal interaction structure, every agent is either strictly stubborn itself (stubborn in the WDTG case) or is temporally connected to a strictly stubborn agent through an influence path (an s-path in the WDTG case). This ensures that stubborn influence permeates the network. Moreover, Lemma~\ref{thm:contraction_defected} shows that such a structure leads to contraction of the state-transition dynamics.
\end{rem}

\begin{lem}
Let \( \Phi(t_d, t_0) \) denote the state transition matrix of the opinion dynamics defined in~\eqref{eq:update_equation} over the interval \([t_0, t_d]\), and let \( \mathcal{G}_{t_0}^{t_d} \) be the corresponding temporal interaction graph. If \( \mathcal{G}_{t_0}^{t_d} \) is a DTG, then the transition matrix \( \Phi(t_d, t_0) \) satisfies the following contraction bound under the \( \ell_1 \)-induced norm:
\[
\|\Phi(t_d, t_0)\| \;\le\; 1 - \epsilon\, w^{\delta}, \qquad \delta := t_d - t_0.
\]
\label{thm:contraction_defected}
\end{lem}

\begin{proof}
By Definition~\ref{defa:dtg_wdtg}, at time \(t_d\) every agent \( \mathrm{v}_i \in \mathcal{V} \) is either (i) strictly stubborn in layer \( \mathcal{G}[t_d] \), or (ii) temporally connected to a strictly stubborn agent through an influential path within \([t_0,t_d]\).  
We analyze these two cases.

\medskip
\noindent \emph{Case 1: strictly stubborn at time \(t_d\).}  
Suppose agent \( \mathrm{v}_i \) is strictly stubborn, so by Definition~\ref{defa:strictly_stubborn} we have
\[
\lambda_i[t_d] \leq 1 - \epsilon.
\]
Since the \(i\)-th row of \( \Lambda[t_d] W[t_d] \) is obtained by multiplying \(\lambda_i[t_d]\) with a convex combination of weights summing to one, the \(i\)-th row sum of the state-transition matrix satisfies
\[
\sum_{j=1}^n \phi_{ij}(t_d,t_0) \leq \lambda_i[t_d] \leq 1 - \epsilon.
\]

\medskip
\noindent \emph{Case 2: connected via an influential path.}  
Suppose agent \( \mathrm{v}_i \) is not itself strictly stubborn, but is connected to a strictly stubborn agent \( \mathrm{v}_j \) at time \(\tau \leq t_d\) through an influential path of length \(\delta_i = t_d-\tau\).  
Along this path each edge weight is at least \(w\), so the total influence transmitted from \(\mathrm{v}_j\) to \(\mathrm{v}_i\) is at least \(w^{\delta_i}\).  Since agent \(\mathrm{v}_j\) is strictly stubborn with parameter \(\epsilon\), the row sum for agent \(i\) is bounded by
\[
\sum_{j=1}^n \phi_{ij}(t_d,t_0) \leq 1 - \epsilon w^{\delta_i}.
\]

\medskip
\noindent \emph{Combining the two cases.}  
In both cases the \(i\)-th row sum is strictly less than one.  
Letting \(\delta := \max_i \delta_i \leq t_d - t_0\), we obtain the uniform contraction bound
\[
\|\Phi(t_d,t_0)\| \;\leq\; 1 - \epsilon w^{\delta}.
\]

Thus, the existence of a defected temporal graph guarantees contraction of the state transition matrix.
\hfill\qedsymbol
\end{proof}

\begin{lem}
\label{thm: union of DTG}
Let \( \Phi(t_d, t_1) \) denote the state transition matrix of the opinion dynamics defined in~\eqref{eq:update_equation} over the interval \([t_1, t_d]\), and let \( \mathcal{G}_{t_1}^{t_d} \) be the corresponding temporal interaction graph. Suppose that this interval is partitioned into \( d \) consecutive subintervals:
\[
\mathcal{G}_{t_1}^{t_2},\ \mathcal{G}_{t_2}^{t_3},\ \dots,\ \mathcal{G}_{t_{d-1}}^{t_d},
\]
among which \( \alpha \) of them are DTGs. Then, the transition matrix \( \Phi(t_d, t_1) \) satisfies the following contraction bound under the \( \ell_1 \)-induced norm:
\[
\|\Phi(t_m,\, t_1)\| \leq \prod_{a=1}^{\alpha}\left(1 - \epsilon w^{\delta_a} \right),
\]
where \( \delta_a = t_{a+1} - t_a \) denotes the length of the \( a \)-th defected interval.
\end{lem}

\begin{proof}
The result is now followed by the sub-multiplicativity of the  \( \ell_1 \)-induced norm. For a sequence of intervals \( [t_1, t_2], [t_2, t_3], \dots, [t_{d-1}, t_d] \), among which \( \alpha \) are defected temporal graph with lengths \( \delta_1, \delta_2, \dots, \delta_\alpha \), we have:
\[
\|\Phi(t_d,\, t_1)\| \leq \prod_{j=1}^{d-1}\|\Phi(t_{j+1},\, t_j)\| \leq \prod_{a=1}^{\alpha}\left(1 - \epsilon w^{\delta_a} \right).
\]
\end{proof}

\begin{lem}
\label{thm: union of WDTG}
Let \( \Phi(t_d, t_0) \) denote the state transition matrix of the opinion dynamics defined in~\eqref{eq:update_equation} over the interval \([t_0, t_d]\), and let \( \mathcal{G}_{t_0}^{t_d} \) be the corresponding temporal interaction graph. Suppose that this interval contains at least one WDTG. Then, the transition matrix satisfies:
\[
\|\Phi(t_d, t_0)\| \;<\; 1.
\]
\end{lem}

\begin{proof}
The proof follows similarly to that of Lemma~\ref{thm: union of DTG}. WDTGs can be viewed as limiting cases of DTGs, in which the degree of stubbornness \( \epsilon \) and the influence weights \( w \) become arbitrarily small. Consequently, the contraction bound established for DTGs also holds for WDTGs, albeit with a non-uniform contraction factor strictly less than one.
\end{proof}

In the TIFJ model, both the influence matrix \( W[t] \) and the susceptibility matrix \( \Lambda[t] \) are constant over time~\cite{friedkin1999social}. Necessary and sufficient conditions for asymptotic stability were derived in~\cite{parsegov2016novel}. These classical results can be interpreted through the lens of WDTGs. The following corollary restates the stability condition for the TIFJ model using this perspective.

\begin{cor}
\label{cor:stationary-stability}
Assume in~\eqref{eq:update_equation} that both \( W[t] = W \) and \( \Lambda[t] = \Lambda \) are constant over time. Then the system is asymptotically stable if and only if the influence graph induced by \( W \) forms a WDTG over some interval \( [0, \tau] \) with \( \tau \leq n \).
\end{cor}

\begin{proof}
The result follows from~\cite{parsegov2016novel}, which establishes that asymptotic stability occurs if and only if each agent is either stubborn or influenced, directly or indirectly, by a stubborn agent. Since the system has \( n \) agents and influence can propagate at most one hop per step, any indirect influence must reach all agents within at most \( n \) steps. This satisfies the definition of a WDTG over some finite interval.
\end{proof}

\section{Boundedness of \(\omega\)-Limit Sets under TVFJ Dynamics}
\label{sec:omegaBoundedness}
The \(\omega\)-limit set collects all accumulation points of the opinion vector \( \mathbf{x}[t] \) as \( t \to \infty \), reflecting the long-term effects of both the initial condition and the evolving influence structure. In this section, we show that the \(\omega\)-limit set of~\eqref{eq:update_equation} is bounded and characterize how asymptotic stability shapes its structure.

\begin{lem}
\label{lem:influence-sum-upper-bound}
Let \( \Phi(t, \tau) \) denote the state transition matrix of the opinion dynamics~\eqref{eq:update_equation}, and let \( \Lambda[\tau] \) be the susceptibility matrix at time \( \tau \). Then, for all \( t \in \mathbb{N} \), the matrix
\[
\Sigma[t] := \sum_{\tau = 0}^{t-1} \Phi(t, \tau + 1)\, (I - \Lambda[\tau])
\]
is row-substochastic.
\end{lem}

\begin{proof}
We proceed by induction on \( t \). For the base case \( t = 1 \),
\[
\Sigma[1] = \Phi(1, 1)\, (I - \Lambda[0]) = I - \Lambda[0],
\]
which is row-substochastic.

Assume \( \Sigma[t] \) is row-substochastic. Then,
\begin{align*}
\Sigma[t+1] &= \sum_{\tau = 0}^{t} \Phi(t+1, \tau + 1)\, (I - \Lambda[\tau]) \\
&= \sum_{\tau = 0}^{t-1} \Phi(t+1, \tau + 1)\, (I - \Lambda[\tau]) + (I - \Lambda[t]) \\
&= \Lambda[t] W[t] \sum_{\tau = 0}^{t-1} \Phi(t, \tau + 1)\, (I - \Lambda[\tau]) + (I - \Lambda[t]) \\
&= \Lambda[t] W[t] \Sigma[t] + (I - \Lambda[t]).
\end{align*}
By the induction hypothesis, \( \Sigma[t] \) is row-substochastic. Since the final expression is a convex combination of the row-substochastic matrix \( W[t]\Sigma[t] \) and the identity matrix, it follows that \( \Sigma[t+1] \) is also row-substochastic.
\end{proof}

\begin{rem}
\label{rem:omega-limit}
Lemma~\ref{lem:influence-sum-upper-bound} guarantees that the general solution~\eqref{eq:fj_solution} remains bounded for all \( t \).
\end{rem}

\begin{cor}
If~\eqref{eq:update_equation} is asymptotically stable, then the corresponding \(\omega\)-limit set \( \omega \) satisfies
\[
\min_i s_i \leq x_i^* \leq \max_i s_i, \quad \forall x_i^* \in \omega.
\]
In particular, all limiting opinions lie within the convex hull of the innate opinion vector \( \mathbf{s} \).
\end{cor}

\begin{proof}
The result follows directly from Lemma~\ref{lem:influence-sum-upper-bound}. Under asymptotic stability, the first term in the general solution~\eqref{eq:fj_solution} vanishes as \( t \to \infty \). Moreover, Lemma~\ref{lem:influence-sum-upper-bound} asserts that
\[
\Sigma[t] := \sum_{\tau=0}^{t-1} \Phi(t, \tau+1)(I - \Lambda[\tau])
\]
is row-substochastic. Thus, the second term of the solution becomes a row-substochastic linear combination of the entries of \( \mathbf{s} \), which implies that all limiting opinions remain within the component-wise bounds of \( \mathbf{s} \).
\end{proof}

As a reference case, in the time-invariant FJ (TIFJ) model asymptotic stability reduces the \(\omega\)-limit set to a singleton~\cite{parsegov2016novel}, with the unique limiting opinion vector given by
\[
\mathbf{x}^* = (I - \Lambda W)^{-1} (I - \Lambda) \mathbf{s}.
\]

\section{Asymptotic Stability of the TVFJ Model}
\label{sec:TVFJStability}

In this section, we extend the algebraic stability conditions of \cite{proskurnikov2017opinion} by introducing graph-theoretic certificates that explicitly capture temporal connectivity. From Corollary~\ref{cor:stationary-stability}, in the time-invariant FJ model, asymptotic stability holds whenever the influence graph induced by \( W \) forms a WDTG over some finite interval \( [0, \tau] \). Equivalently, stability requires that the temporal interaction graph of the system contains infinitely many WDTGs. However, as the following example shows, this condition is not sufficient in the time-varying case.

\begin{ex}
\label{ex: WDTG is not sufficient for contraction}
Consider a three-agent network with a fixed row-stochastic influence matrix
\[
W = \frac{1}{3}
\begin{bmatrix}
1 & 1 & 1 \\
1 & 1 & 1 \\
1 & 1 & 1
\end{bmatrix}
= \tfrac{1}{3} \mathbf{1}\mathbf{1}^\top,
\]
and a time-varying susceptibility matrix \( \Lambda[t] = \lambda[t] I \), where
\[
\lambda[t] := 1 - \frac{1}{(t+1)^2}, \qquad t \in \mathbb{N}_0.
\]

The update matrix at each step is \( \Gamma[t] = \lambda[t] W \), so the state transition matrix is
\[
\Phi(t, 0) = \left( \prod_{k=0}^{t-1} \lambda[k] \right) W = \Lambda_t W,
\]
with
\[
\Lambda_t := \prod_{k=0}^{t-1} \left(1 - \tfrac{1}{(k+1)^2} \right).
\]

Since each \( \lambda[t] < 1 \), all agents are stubborn at all times. Moreover, because \( W \) is fully connected, the temporal graph \( \mathcal{G}[t] \) forms a WDTG at every step.

However, the deviation from full susceptibility is summable:
\[
\sum_{t=0}^\infty (1 - \lambda[t]) = \sum_{t=0}^\infty \tfrac{1}{(t+1)^2} = \tfrac{\pi^2}{6} < \infty,
\]
so the product \( \Lambda_t \) converges to a constant strictly greater than zero:
\[
\lim_{t \to \infty} \Lambda_t > 0.
\]
Consequently,
\[
\lim_{t \to \infty} \|\Phi(t, 0)\| = \tfrac{1}{3} \cdot \lim_{t \to \infty} \Lambda_t > 0,
\]
and the system fails to be asymptotically stable.
\end{ex}

As highlighted in Example~\ref{ex: WDTG is not sufficient for contraction}, the infinite product of contractions that are arbitrarily close to identity may not converge to zero. Hence, contraction alone is insufficient; the contractions must be \emph{uniformly bounded away from one} to guarantee the decay of the influence of the initial condition. This motivates the development of stronger structural conditions on the temporal interaction graph. In what follows, we present a \emph{graph-theoretic} sufficient condition for asymptotic stability, showing that the existence of infinitely many defected temporal graphs (DTGs) guarantees convergence.

\begin{thm}
\label{thm: asymptotic stability infinite DTGs}
Consider the dynamics in~\eqref{eq:update_equation}. If there exist infinitely many pairwise-disjoint finite intervals \( [t_i, t_{i+1}) \) such that each temporal graph \( \mathcal{G}_{t_i}^{t_{i+1}} \) is a \emph{defected temporal graph}, then the system is asymptotically stable.
\end{thm}

\begin{proof}
Let the interval \( [0, t] \) be partitioned into \( m(t) \) disjoint DTG intervals and \( r(t) \) non-defected intervals. Denote the length of the \( i \)-th DTG interval by \( \delta_i = t_{i+1} - t_i \). By Lemma~\ref{thm: union of DTG}, the following inequality holds:
\begin{align}
\|\Phi(t+1, 0)\| \leq \prod_{k=1}^{r(t)} \|\Phi(t_k, t_{k+1})\| \cdot \prod_{i=1}^{m(t)} \bigl(1 - \epsilon w^{\delta_i} \bigr).
\label{eq: convergence equation}
\end{align}
Each non-defected segment satisfies \( \|\Phi(t_k, t_{k+1})\| \leq 1 \), while each DTG segment contributes a contraction factor strictly less than one. Since \( m(t) \to \infty \) as \( t \to \infty \), it follows that
\(
\lim_{t \to \infty} \|\Phi(t+1, 0)\| = 0,
\)
which establishes asymptotic stability.
\end{proof}

\begin{rem}
The asymptotic stability condition in~\cite{proskurnikov2017opinion} is closely related to our DTG-based theorem, but differs in abstraction. While~\cite{proskurnikov2017opinion} states stability via algebraic constraints on \(\{\Lambda[t],W[t]\)\,), our formulation encodes it \emph{graph-theoretically} through the temporal interaction topology. In particular, the DTG sufficient condition requires that, at some time layer, every agent is either strictly stubborn or temporally connected to a strictly stubborn agent via an influential path. This yields a transparent, interpretable view of influence propagation and helps identify and visualize stability-critical patterns in large-scale networks.
\end{rem}

\begin{ex}
Consider a five-agent opinion dynamics model with two alternating interaction networks:
\begin{align*}
A_1 = 
\begin{bmatrix}
1 & 0 & 0 & 0 & 0 \\
0 & 1 & 0 & 1 & 0 \\
0 & 0 & 1 & 0 & 1 \\
0 & 1 & 0 & 1 & 0 \\
0 & 0 & 1 & 0 & 1 
\end{bmatrix}, \quad
A_2 = 
\begin{bmatrix}
1 & 0 & 0 & 0 & 0 \\
1 & 0 & 1 & 0 & 0 \\
0 & 1 & 1 & 0 & 0 \\
0 & 0 & 0 & 1 & 1 \\
0 & 0 & 0 & 1 & 1
\end{bmatrix}.
\end{align*}

Each interaction matrix is paired with the same innate opinion vector and trust weight matrix:
\[
\mathbf{s} = 
\begin{bmatrix}
0.5 \\ 1.0 \\ 1.0 \\ 0.0 \\ 0.0
\end{bmatrix}, \quad
\hat{W} = 
\begin{bmatrix}
1       & 0       & 0       & 0       & 0 \\
\tfrac{1}{4} & \tfrac{1}{4} & \tfrac{1}{4} & \tfrac{1}{4} & 0 \\
0       & \tfrac{1}{3} & \tfrac{1}{3} & 0       & \tfrac{1}{3} \\
0       & \tfrac{1}{3} & 0       & \tfrac{1}{3} & \tfrac{1}{3} \\
0       & 0       & \tfrac{1}{3} & \tfrac{1}{3} & \tfrac{1}{3}
\end{bmatrix}.
\]

The actual weight matrix at time \( t \) is
\[
w_{ij}[t] := \frac{a_{ij}[t]\, \hat{w}_{ij}}{\sum_{k \in \mathcal{N}_i[t]} a_{ik}[t]\, \hat{w}_{ik}},
\]
where \( A[t] \in \{A_1, A_2\} \) and \( \mathcal{N}_i[t] \) is the neighbor set of agent \( v_i \) at time \( t \). We specify:
\begin{itemize}
    \item \textbf{Network 1:} Agent \( v_1 \) is completely isolated, while the remaining agents interact in fixed local patterns. The susceptibility matrix is
    \[
    \Lambda_1 = \text{diag}(0,~1,~1,~1,~1).
    \]

    \item \textbf{Network 2:} Agent \( v_1 \) acts as a strictly stubborn external source influencing agent \( v_2 \). The susceptibility values vary with time as
    \[
    \Lambda_2 = \text{diag}(0,~\lambda_2[t],~1,~\lambda_4[t],~1),
    \]
    where \( \lambda_2[t] = 0.9 - \tfrac{1}{1+k} \), \( \lambda_4[t] = 1 - \tfrac{1}{1+k} \), and \( k \) denotes the number of switches into Network~2 up to time \( t \).
\end{itemize}

Agents switch to Network~2 at times
\[ t_k = t_{k-1} + \lfloor \log(t_{k-1}+1) \rfloor +d\]

interacting for a fixed duration \( d \), where \( \lfloor \cdot \rfloor \) denotes the floor operator. The opinion dynamics follow the TVFJ model:
\[
\mathbf{x}[t+1] = \Lambda_{\sigma(t)} W_{\sigma(t)} \mathbf{x}[t] + \bigl(I - \Lambda_{\sigma(t)}\bigr)\mathbf{s},
\]
where \( \sigma(t) \in \{1, 2\} \) indicates the active network.

As \( t, t_k \to \infty \), the dynamics increasingly resemble the French--DeGroot model, which converges to a consensus depending on the initial condition. Figure~\ref{fig: Example AS TVFJ Fig1} shows consensus-like behavior among agents \( v_2 \) through \( v_4 \) for \( d = 2 \). However, global consensus is not reached due to influence propagating through the strictly stubborn agent \( v_2 \).

Figure~\ref{fig: Example AS TVFJ Graph} illustrates that \( \mathcal{G}_{t_k + d}^{t_k + d + 1} \) forms a DTG at every switching time \( t_k \), which guarantees asymptotic stability. Consequently, the \(\omega\)-limit set becomes independent of the initial condition. Figure~\ref{fig: Example AS TVFJ Fig2} shows zero-input opinion trajectories, \(\Phi(t, t_0)\,\mathbf{x}[t_0]\), under different initial conditions. All trajectories converge to zero, confirming that the long-term behavior depends solely on the model parameters \( \mathbf{s} \), \( \Lambda[t] \), and \( W[t] \), rather than on the initial state.

\begin{figure}[h]
\centering
\includegraphics[width=0.35\textwidth]{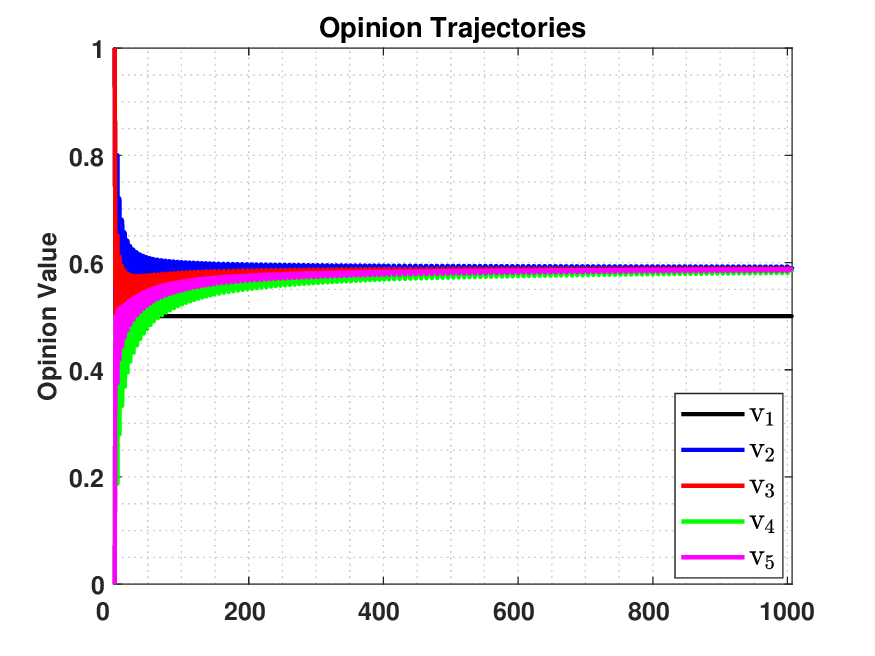}
\caption{Opinion trajectories of agents \( v_1 \) to \( v_5 \) under alternating networks with fixed Network~2 duration \( d = 2 \). The dynamics resemble French--DeGroot consensus due to increasing switching intervals.}
\label{fig: Example AS TVFJ Fig1}
\end{figure}

\begin{figure}[h]
\centering
\includegraphics[width=0.25\textwidth]{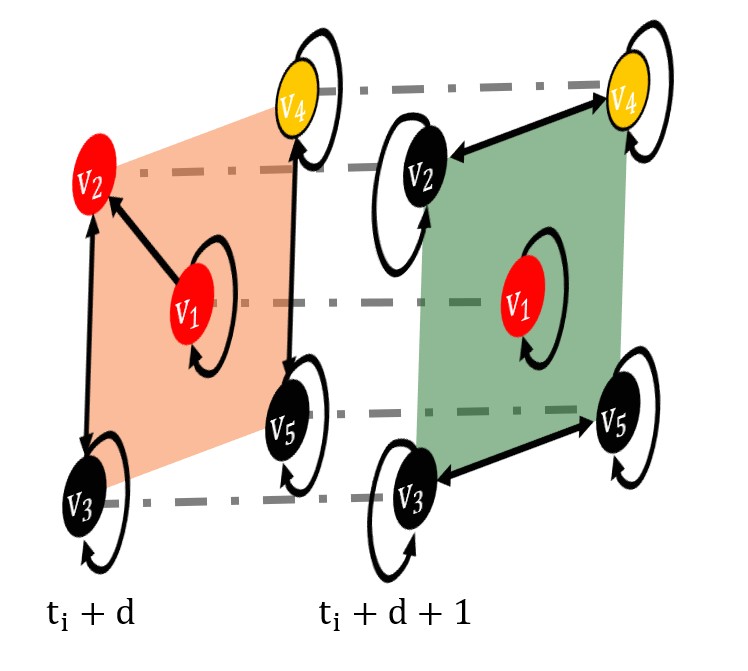}
\caption{Interaction topology at time \( t_k + d \), where \( \mathcal{G}_{t_k + d}^{t_k + d + 1} \) forms a DTG. The presence of this DTG in every cycle guarantees asymptotic stability, regardless of initial conditions.}
\label{fig: Example AS TVFJ Graph}
\end{figure}

\begin{figure}[h]
\centering
\includegraphics[width=0.4\textwidth]{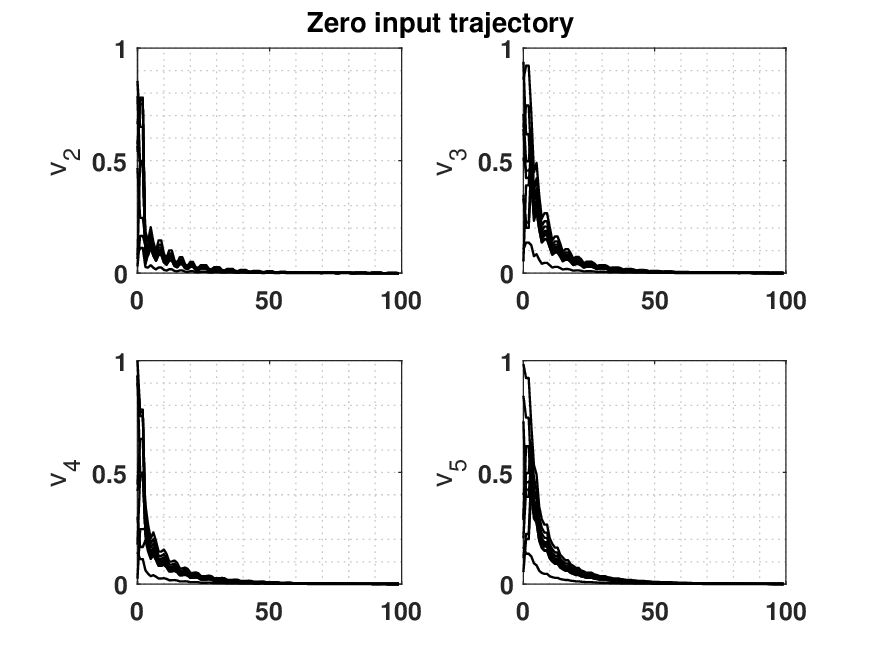}
\caption{Zero-input opinion trajectories under different initial conditions. All trajectories converge to zero, showing asymptotic stability and independence from initial conditions.}
\label{fig: Example AS TVFJ Fig2}
\end{figure}
\label{ex: Example AS TVFJ}
\end{ex}

\subsection{Semi-Periodic Defected Networks}

Consider a recommender system that persistently exposes agents to influential content, thereby shaping public opinion~\cite{bakshy2015exposure}. In this setting, it is natural to assume the existence of a fixed integer \( p_s > 0 \) such that, for every \( t \), the temporal graph \( \mathcal{G}_t^{t+p_s} \) forms a defected temporal graph. Opinion dynamics under this condition are termed a \emph{semi-periodic defected network}.

\begin{thm}
Let the opinion dynamics system~\eqref{eq:update_equation} evolve under a semi-periodic defected network with period length \( p_s \). Then, the system is \emph{exponentially stable} with growth factor \( c \) and exponential decay rate \( \gamma \) given by:
\[
c = \frac{1}{(1 - \varepsilon w^{p_s})^{p_s}}, \quad \text{and} \quad \gamma \leq \sqrt[p_s]{1 - \varepsilon w^{p_s}}.
\]
\label{thm: UES under semi-periodic}
\end{thm}

\begin{proof}
Let \( t - \tau = q p_s + r \), where \( q \in \mathbb{N} \), \( 0 \leq r < p_s \). Then:
\[
\Phi(t, \tau) = \Phi(t, \tau + q p_s) \prod_{k=0}^{q-1} \Phi(\tau + (k+1)p_s,\, \tau + k p_s).
\]

By the semi-periodic assumption, each interval \( [\tau + k p_s,\, \tau + (k+1) p_s] \) contains at least one defected layer. Therefore, by Lemma~\ref{thm:contraction_defected}, each product term satisfies:
\[
\|\Phi(\tau + (k+1)p_s,\, \tau + k p_s)\| \leq 1 - \varepsilon w^{p_s}.
\]

Since the product norm is submultiplicative and \(\|\Phi(t, \tau + q p_s)\| \leq 1\), we get:
\[
\|\Phi(t, \tau)\| \leq (1 - \varepsilon w^{p_s})^q.
\]

Since \( q = \frac{t - \tau -r}{p_s} \), and \(0 \leq r < p_s\) we conclude:
\[
\|\Phi(t, \tau)\|\leq \frac{1}{(1 - \varepsilon w^{p_s})^{p_s}}\left(\sqrt[p_s]{1 - \varepsilon w^{p_s}}\right)^{t-\tau}.
\]

Hence, the exponential stability condition is satisfied with:
\[
c = \frac{1}{(1 - \varepsilon w^{p_s})^{p_s}}, \quad \gamma = \sqrt[p_s]{1 - \varepsilon w^{p_s}}.
\]
\end{proof}

\section{Asymptotic Stability of Trust–Based Opinion Dynamics}
\label{sec:TrustBasedStability}
Theorem~\ref{thm: asymptotic stability infinite DTGs} shows that the TVFJ model is asymptotically stable when the temporal graph contains infinitely many DTGs. We now prove that the trust–based model of Definition~\ref{defa:Trust_Based_TVFJ} attains asymptotic stability under a strictly weaker requirement: the existence of infinitely many WDTGs. Intuitively, the trust-based update induces contraction even when the temporal network is connected only via \(s\)-paths.

\begin{thm}
\label{thm:asymptotic_stability_TBFJ}
Consider the TBFJ opinion dynamics in Definition~\ref{defa:Trust_Based_TVFJ}. 
If there exist infinitely many disjoint finite intervals \( [t_i,t_{i+1}) \) such that each temporal graph \( \mathcal{G}_{t_i}^{t_{i+1}} \) is a \emph{weakly defected temporal graph}, and and moreover, the length of each such interval is uniformly bounded, i.e., 
\[t_{i+1}-t_i \leq \delta \quad \text{for some finite \(\delta >0\)}\]
Then the system is asymptotically stable.
\end{thm}

\begin{proof}
Let \( n \) be the number of agents. Since \( A[t] \in \{0,1\}^{n \times n} \), the number of distinct adjacency matrices is finite. Given the fixed trust matrix \( \hat{W} \), each adjacency matrix \( A[t] \) determines a unique influence matrix \( W[t] \), and each agent computes its susceptibility \( \lambda_i[t] \in [0,1] \) based on its reachable neighbors.

Define the set
\[
\mathcal{A} := \left\{ \{A[k]\}_{k=t_i}^{t_{i+1}-1} \,\middle|\, \mathcal{G}_{t_i}^{t_{i+1}} \text{ is a WDTG} \right\},
\]
that contains all uniformly bounded length sequences of adjacency matrices such that the resulting temporal interaction graph over that interval is a weakly defected temporal graph. It is clear that \( \mathcal{A} \) is a countable and bounded set.

Let \( b = |\mathcal{A}| \) denote the (finite) cardinality of the set \( \mathcal{A} \). According to Lemma ~\ref{thm: union of WDTG}, for each sequence \(  \{A[k]\}_{k=t_i}^{t_{i+1}-1} \in \mathcal{A} \), there exists a constant \( r_i \in (0,1) \) such that the corresponding state transition matrix in the interval satisfies
\[
\|\Phi(t_{i+1}, t_i)\| \leq r_i.
\]

Let \( \alpha_i(t) \) denote the number of occurrences of the \(i\)-th sequence in \( \mathcal{A} \) up to time \( t \). Then for all \( t \),
\[
\|\Phi(t, 0)\| \leq r_1^{\alpha_1(t)} \cdot r_2^{\alpha_2(t)} \cdots r_b^{\alpha_b(t)}.
\]

Since the sequence of time intervals \( [t_i, t_{i+1}) \) is infinite and the number of distinct sequences is finite, it follows that for some \( i \), \( \alpha_i(t) \to \infty \) as \( t \to \infty \). Therefore,
\[
\lim_{t \to \infty} \|\Phi(t, 0)\| = 0.
\]
\end{proof}

The time-invariant FJ (TIFJ) and periodically switching FJ (PSFJ) models are special cases of the trust-based opinion dynamics framework. Under asymptotic stability, the \( \omega \)-limit set of TIFJ is fully determined by the model parameters. In what follows, we show that the same property holds for PSFJ models as well.

\subsection{Asymptotic Behavior of the PSFJ Model}

A particularly structured and analytically tractable subclass of TBFJ models emerges when the adjacency matrix \( A[t] \) exhibits a periodic pattern with a fixed period \( p \). Formally, we assume there exists a finite sequence of adjacency matrices \(\{A_l\}_{l=1}^p\) repeating periodically over time, i.e.,
\[
A[t] = A_{\langle t \rangle_p}, \quad \text{with} \quad \langle t \rangle_p = (t \bmod p) + 1.
\]

According to Definition~\ref{defa:Trust_Based_TVFJ}, each adjacency matrix \( A_l \), for \( l \in \{1,\dots,p\} \), corresponds uniquely to a pair \(\{W_l, \Lambda_l\}\) that characterizes the interaction weights and susceptibility profiles during the \( l \)-th phase of each periodic cycle. Under these assumptions, the opinion dynamics given by~\eqref{eq:update_equation} simplify to a periodically switching model:
\begin{equation}
    \mathbf{x}[t+1] = \Lambda_{\langle t \rangle_p} W_{\langle t \rangle_p} \mathbf{x}[t] + \left(I - \Lambda_{\langle t \rangle_p}\right)\mathbf{s}.
    \label{eq: periodic update equation}
\end{equation}

A key advantage of this periodic structure is that it allows the time-varying dynamics~\eqref{eq: periodic update equation} to be decomposed into a collection of \( p \) distinct linear time-invariant (LTI) subsystems.

\begin{thm}
\label{thm: LTI representation}
Consider the PSFJ opinion dynamics~\eqref{eq: periodic update equation} with period \( p \). This system can be decoupled into \( p \) linear time-invariant subsystems. Specifically, for each \( l \in \{1, 2, \dots, p\} \), define the subsystem state \(\mathbf{x}_l[k]\) at scaled discrete time \( k = \frac{t}{p} \). Then, each subsystem evolves according to:
\begin{equation}
    \mathbf{x}_l[k+1] = M_l\,\mathbf{x}_l[k] + N_l\,\mathbf{s},
    \label{eq: TIL model}
\end{equation}
where
\begin{align}
    M_l &= \prod_{j=0}^{p-1} \Lambda_{\langle l+j\rangle_p}\,W_{\langle l+j\rangle_p}, 
    \label{eq: TIL model_C3}\\[6pt]
    N_l &= \sum_{j=0}^{p-1}\left(\prod_{r=j+1}^{p-1} \Lambda_{\langle l+r\rangle_p}\,W_{\langle l+r\rangle_p}\right)\left(I - \Lambda_{\langle l+j\rangle_p}\right).
    \label{eq: TIL model_C4}
\end{align}
\end{thm}

\begin{proof}
Let \(\mathbf{x}_l[t]\) denote the state trajectory starting at time step \( t \) under the switching mode initiated by the pair \( \{W_l, \Lambda_l\} \). Expanding the dynamics explicitly over one full period yields:
\begin{align*}
    &\mathbf{x}_l[t+1] = \Lambda_l W_l \mathbf{x}_l[t] + (I-\Lambda_l)\mathbf{s},\\
    &\mathbf{x}_l[t+2] = \Lambda_{\langle l+1 \rangle_p} W_{\langle l+1 \rangle_p} \mathbf{x}_l[t+1] + (I-\Lambda_{\langle l+1 \rangle_p})\mathbf{s},\\
    &\hspace{4cm}\vdots\\[-5pt]
    &\mathbf{x}_l[t+p] = \Lambda_{\langle l-1 \rangle_p} W_{{\langle l-1 \rangle_p}}\mathbf{x}_l[t+p-1]+(I-\Lambda_{{\langle l-1 \rangle_p}})\mathbf{s}.
\end{align*}
Using the explicit solution~\eqref{eq:fj_solution}, we compactly relate \(\mathbf{x}_l[t+p]\) to \(\mathbf{x}_l[t]\) as:
\[
\mathbf{x}_l[t+p] = {M}_l\mathbf{x}_l[t]+{N}_l\mathbf{s},
\]
where matrices \({M}_l\) and \({N}_l\) are defined in \eqref{eq: TIL model_C3} and \eqref{eq: TIL model_C4}, respectively.

By introducing the scaled time variable \( k = \frac{t}{p} \), the dynamics equivalently become:
\[
\mathbf{x}_l[k+1] = {M}_l \mathbf{x}_l[k] + {N}_l\mathbf{s}.
\]
This explicitly demonstrates the decomposition into \( p \) LTI subsystems.
\end{proof}

\begin{thm}
\label{thm: exponential stability of PSFJ}
Consider the opinion dynamics~\eqref{eq: periodic update equation} with period \( p \). Suppose the temporal graph \( \mathcal{G}_0^{p} \) is a WDTG. Then, the system is exponentially stable, with constants of growth factor \( c \) and exponential decay rate \( \gamma \) 
\[
c = \frac{1}{\|\Phi(p,0)\|^p}, \quad \gamma = \sqrt[p]{\|\Phi(p,0)\|}.
\]

Moreover, the \(\omega\)-limit set contains at most \( p \) points, independent of initial conditions. Each asymptotic point is explicitly characterized by:
\[
\mathbf{x}_l^* = (I - {M}_l)^{-1} {N}_l \, \mathbf{s}, \quad l = 1, 2, \dots, p,
\]
where \( M_l \) and \( N_l \) are defined as in Theorem~\ref{thm: LTI representation}.
\end{thm}

\begin{proof}
Given that the switching occurs periodically with period \( p \), Lemma \ref{thm: union of WDTG} indicates that the transition matrix adheres to:
\[
\|\Phi(t+p,t)\| = \|\Phi(p,0)\| < 1, \quad \forall t.
\]

Letting \(t - \tau = qp + r\), where \(q \geq 0\) and \(0 \leq r < p\), we have:
\[
\|\Phi(t,\tau)\| \leq \|\Phi(p,0)\|^q \leq \frac{1}{\|\Phi(p,0)\|^p}\left(\sqrt[p]{\|\Phi(p,0)\|}\right)^{t-\tau}.
\]

Thus, ES follows with explicitly given constants \(c\) and \(\gamma\). Consequently, each subsystem (Theorem~\ref{thm: LTI representation}) converges uniquely to its fixed point, establishing the claimed \(\omega\)-limit set structure.
\end{proof}

\begin{ex}
Consider Example~\ref{ex: Example AS TVFJ} with the following modifications, yielding a periodically switching structure. Let the switching times be defined by \( t_k = t + 2 \), resulting in a switching period \( p = 4 \). Assume the susceptibility matrix in \textbf{Network 2} is \(\Lambda_2= \text{diag}(0,~0.9,~1,~1,~1) \). Under this setting, the system becomes a PSFJ model. Each subsystem converges to the same fixed point:
\[
\mathbf{x}_l^* = \begin{bmatrix}
0.50,~ 0.59,~ 0.59,~ 0.59,~ 0.59
\end{bmatrix}^ \top.
\]
The \(\omega\)-limit set contains a single point. The resulting opinion trajectory exhibits consensus-like behavior with reduced fluctuations, as shown in Figure.~\ref{fig: Example AS PSFJ Fig1}, closely resembling Figure.~\ref{fig: Example AS TVFJ Fig1}.

Now consider the case where \( \lambda_2 = \lambda_4 = 0.9 \). In this scenario, the \(\omega\)-limit set contains three points. This multi-point behavior is illustrated in Figure.~\ref{fig: Example AS PSFJ Fig2}, in accordance with Theorem~\ref{thm: exponential stability of PSFJ}.

\label{ex: Periodic Switching FJ Model}
\end{ex}

\begin{figure}[h]
\centering
\includegraphics[width=0.35\textwidth]{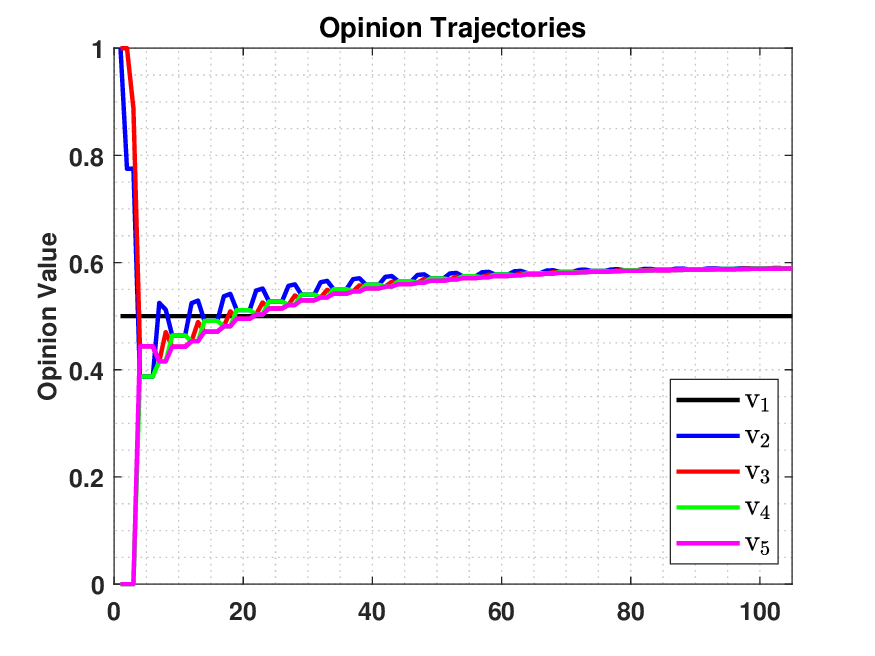}
\caption{Opinion trajectories of agents \( \mathrm{v}_1 \) to \( \mathrm{v}_5 \) in the periodic switching scenario with \(\Lambda_2= \text{diag}(0,~0.9,~1,~1,~1) \).}
\label{fig: Example AS PSFJ Fig1}
\end{figure}

\begin{figure}[h]
\centering
\includegraphics[width=0.35\textwidth]{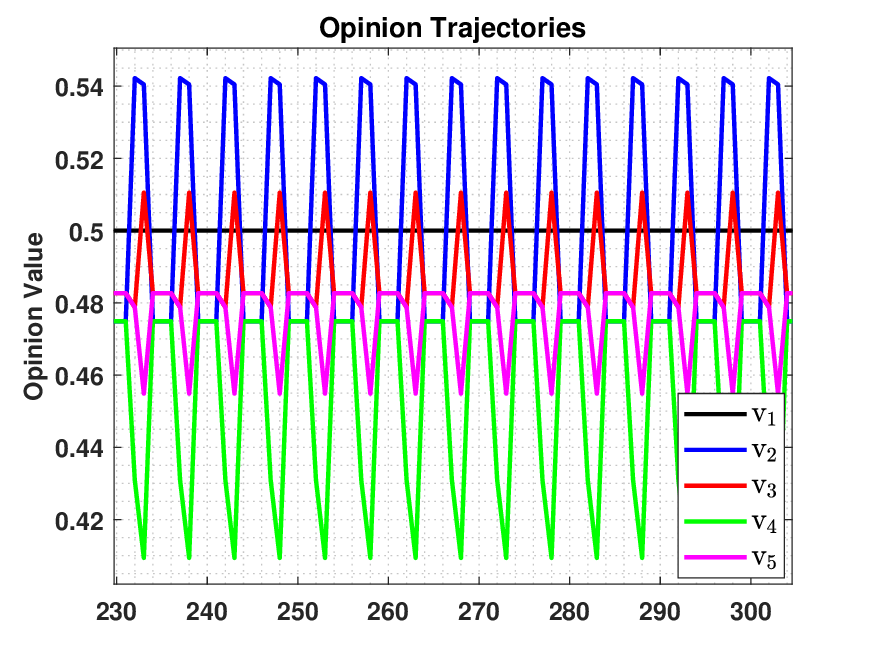}
\caption{Opinion trajectories of agents \( \mathrm{v}_1 \) to \( \mathrm{v}_5 \) in the periodic switching scenario with \(\Lambda_2= \text{diag}(0,~0.9,~1,~0.9,~1) \).}
\label{fig: Example AS PSFJ Fig2}
\end{figure}

\section{Robustness of Stability}
\label{sec: RobustnessofStabilityUnderPerturbations}

In practice, exact interaction matrices are rarely known; instead, models rely on estimates that may include noise or structural errors. Building on the exponential stability results for semi-periodic (Theorem~\ref{thm: UES under semi-periodic}) and periodic switching networks (Theorem~\ref{thm: exponential stability of PSFJ}), we now establish that these guarantees persist under bounded perturbations.

Consider the perturbed dynamics:
\begin{align}
\mathbf{x}[t+1] = P[t] \mathbf{x}[t] + D[t] \mathbf{s},
\label{eq: Perturbed Dynamics Model}
\end{align}
where \( P[t] \in \mathbb{R}^{n \times n} \) and \( D[t] \in \mathbb{R}^{n \times n} \) have nonnegative entries, with \( D[t] \) diagonal, such that
\begin{align}
    \bigl(P[t] + D[t]\bigr)\mathbf{1}_n = \mathbf{1}_n.
    \label{eq: Perturbed Dynamics Model row-stochastic condition}
\end{align}

Assume there exist nominal matrices \( \bar{W}[t] \) and \( \bar{\Lambda}[t] \), along with a bounded perturbation \( E[t] \), such that
\[
P[t] = \bar{\Lambda}[t]\bar{W}[t] + E[t],
\]
with:
\begin{itemize}
    \item \( \bar{W}[t] \) row-stochastic,
    \item \( \bar{\Lambda}[t] \) diagonal with entries in \([0,1]\),
    \item \( E[t] \) uniformly bounded in norm.
\end{itemize}
The corresponding nominal TVFJ model is
\begin{align}
\bar{\mathbf{x}}[t+1] = \bar{\Lambda}[t] \bar{W}[t] \bar{\mathbf{x}}[t] + \bigl(I - \bar{\Lambda}[t]\bigr)\mathbf{s}.
\label{eq: nominal TVFJ}
\end{align}

\begin{rem}
From~\eqref{eq: Perturbed Dynamics Model row-stochastic condition},
\[
d_i[t] = 1 - \bar{\lambda}_i[t] - \sum_{j=1}^n e_{ij}[t],
\]
where \( \bar{\lambda}_i[t] \) and \( d_i[t] \) are the \( i \)th diagonal entries of \( \bar{\Lambda}[t] \) and \( D[t] \), respectively.  
If \(\sum_{j=1}^n e_{ij}[t] = 0\), the perturbation affects only the interaction structure without altering agent \( \mathrm{v}_i \)’s susceptibility. In this case, the agent adjusts relative neighbor weights (e.g., due to local observations or shifts in interpersonal trust) while maintaining its overall openness to influence.
\end{rem}

\begin{thm}
\label{thm: perturbed-UES}
Assume the nominal model~\eqref{eq: nominal TVFJ} is exponentially stable with growth factor \( c \) and decay rate \( \gamma \). If the perturbations satisfy
\[
\|E[t]\| < -\frac{\gamma}{c} \ln \gamma,
\]
then the perturbed dynamics~\eqref{eq: Perturbed Dynamics Model} are also  exponentially stable. In particular:
\begin{itemize}
    \item If the SPFJ conditions of Theorem~\ref{thm: UES under semi-periodic} hold with period \( p_s \), then ES holds if
    \[
    \|E[t]\| < -\frac{1}{p_s}(1 - \varepsilon w^{p_s})^{p_s + \frac{1}{p_s}} \ln (1 - \varepsilon w^{p_s}),
    \]
    \item If the PSFJ conditions of Theorem~\ref{thm: exponential stability of PSFJ} hold with period \( p \), then ES holds if
    \[
    \|E[t]\| < -\frac{1}{p} \|\Phi(p,0)\|^{p + \tfrac{1}{p}} \ln \|\Phi(p,0)\|.
    \]
\end{itemize}
\end{thm}

\begin{proof}
The result follows from a robustness theorem for discrete-time linear time-varying systems~\cite{zhou2017asymptotic}:

\begin{lem}
\label{lem: perturbed-ES}
Consider the perturbed system
\[
\mathbf{x}[t+1] = \bigl(A[t] + E[t]\bigr)\mathbf{x}[t],
\]
where the nominal system \( A[t] \) is ES with growth factor \( c \) and decay rate \( \gamma \in (0,1) \).  
If
\[
\sum_{k = t_0}^{t-1} \|E[k]\| \leq \zeta (t - t_0) + \beta(t_0, \zeta), \quad \forall t \geq t_0,
\]
for some \( \beta(t_0, \zeta) > 0 \), and
\[
\zeta < -\tfrac{\gamma}{c}\ln \gamma,
\]
then the perturbed system is asymptotically stable. Moreover, if \( \beta(t_0, \zeta) \) is independent of \( t_0 \), the system is exponentially stable.
\end{lem}

Applying Lemma~\ref{lem: perturbed-ES} to~\eqref{eq: Perturbed Dynamics Model}, we require
\[
\|E[t]\| < -\tfrac{\gamma}{c} \ln \gamma.
\]

For the semi-periodic case (Theorem~\ref{thm: UES under semi-periodic}),
\[
c = \frac{1}{(1 - \varepsilon w^{p_s})^{p_s}}, \quad 
\gamma = \sqrt[p_s]{1 - \varepsilon w^{p_s}},
\]
which yields the claimed bound.  
The periodic case follows analogously by replacing \( p_s \) with \( p \) and using the UES parameters of Theorem~\ref{thm: exponential stability of PSFJ}.
\end{proof}

DTGs and WDTGs naturally arise in real networks, especially small-world or structured settings, making PSFJ and SPFJ models realistic rather than idealized~\cite{watts1998collective}. However, stability guarantees for perfectly periodic systems do not directly cover perturbed networks. Theorem~\ref{thm: perturbed-UES} fills this gap by showing that qualitative convergence persists under bounded deviations, extending stability to a broad class of practical, imperfect networks.

\section{Conclusion}
\label{Conclution}

We presented a graph–theoretic framework for analyzing the stability of opinion dynamics in the time-varying Friedkin–Johnsen model. Central to this framework are two novel temporal structures, \emph{defected temporal graphs} and \emph{weakly defected temporal graphs}, which act as topological certificates linking temporal connectivity and stubborn influence to contraction of the state transition matrix. These structures yield several core results:
\begin{enumerate}
\item \emph{Asymptotic stability} of the TVFJ model whenever DTGs recur infinitely often,
\item \emph{Exponential stability} in SPFJ, with explicit growth and decay constants,
\item \emph{Asymptotic stability} TBFJ models under the weaker requirement of infinitely many WDTGs, and
\item \emph{Bounded $\omega$-limit sets}, ensuring that all long-run opinions remain within the convex hull of the innate opinion vector.
\end{enumerate}

For periodically switching systems, we further decomposed the dynamics into $p$ LTI subsystems, enabling closed-form characterizations of the $\omega$-limit set and explicit stability certificates. These results show that periodic or semi-periodic recurrence of defected layers guarantees both convergence and predictability of long-run outcomes.

In practice, however, network structures are rarely known exactly, and real systems are subject to noise, estimation error, and temporal irregularities. To address this, we established a \emph{robustness theorem} demonstrating that \emph{exponential stability} persists under bounded perturbations of the interaction weights. This bridges the gap between idealized periodic models and realistic, imperfect networks, ensuring that qualitative asymptotic behavior is preserved even under structural or parametric deviations. A broad class of practical networks, including those approximating periodic cycles or subject to small perturbations, thus falls within the stability regime guaranteed by our analysis.

Beyond formal guarantees, the framework offers a principled lens on \emph{human–AI interaction in evolving social networks}. Human agents typically update their opinions with bounded rationality, guided by heuristics and local trust, while AI agents, such as recommender systems, inject globally optimizing rules that make the entire network highly time-varying. The DTG/WDTG framework highlights these asymmetries: stubborn influence and temporal connectivity can be visualized directly in the graph structure, providing interpretable and scalable stability certificates for complex, evolving networks.

Overall, this work establishes a rigorous, graph-theoretic foundation for the study of time-varying opinion dynamics, complementing algebraic approaches centered on ${\Lambda[t], W[t]}$. It shows that stability can be certified directly from temporal interaction topology, even under variability and perturbations. Looking ahead, several directions merit further study: extending robustness guarantees to broader uncertainty models, analyzing multi-dimensional opinion spaces, incorporating strategic game-theoretic behavior, and modeling richer mechanisms of trust formation and adaptation. These developments will deepen our understanding of how human and AI agents jointly shape collective outcomes in dynamic networks.

\bibliographystyle{IEEEtran}
\bibliography{ref}

\begin{IEEEbiography}[{\includegraphics[width=1.1in,height=1.25in,clip]{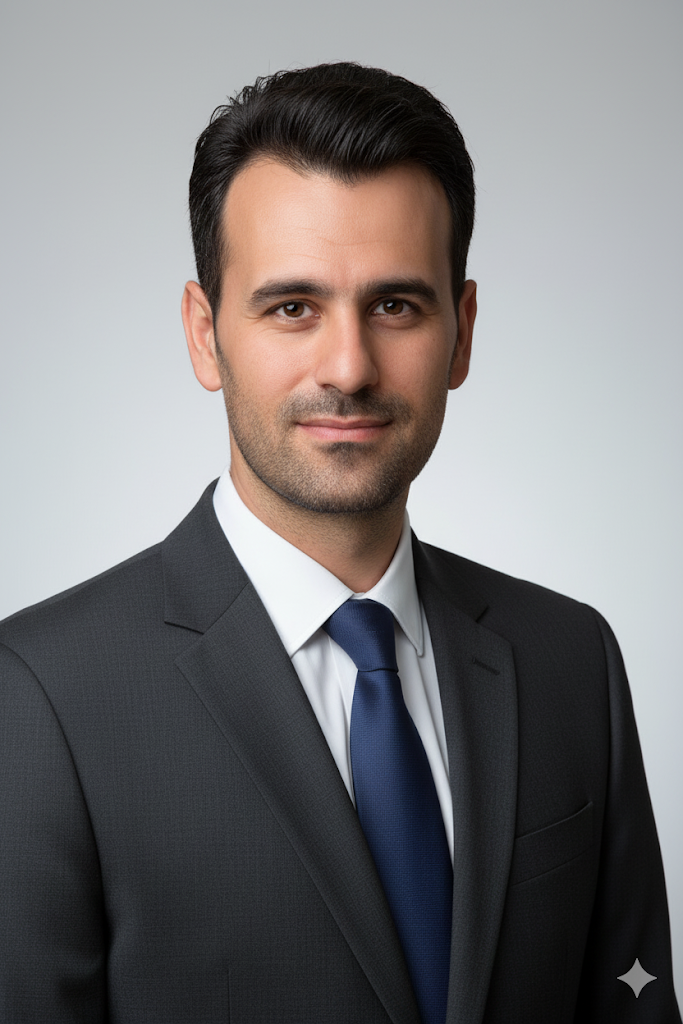}}]{M. Hossein Abedinzadeh}
is currently pursuing a Ph.D. in the Department of Electrical and Computer Engineering at the State University of New York in Binghamton, NY. From 2018 to 2021, he served as a research associate at Isfahan University of Technology, where he worked on modeling, simulation, and control of complex dynamical systems. He earned his second M.Sc. in Applied Mathematics: Differential Equations and Dynamical Systems from IUT in 2020. His first M.Sc. was in Aerospace Engineering: Flight Dynamics and Control from Sharif University of Technology, Tehran, Iran, in 2011. His research interests mainly focus on complex social network analysis.
\end{IEEEbiography}

\begin{IEEEbiography}[{\includegraphics[width=1.1in,height=1.25in,clip]{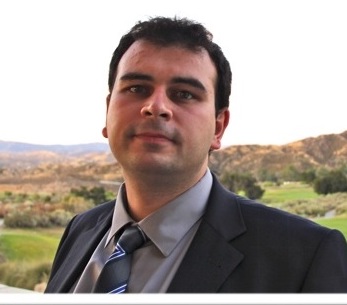}}]{Emrah Akyol}
is an Associate Professor of Electrical and Computer Engineering at the State University of New York at Binghamton, where he has been since Sept. 2017. He received his Ph.D. degree in 2011 from the University of California at Santa Barbara. From 2006 to 2007, he held positions at Hewlett–Packard Laboratories and NTT Docomo Laboratories, both in Palo Alto, CA, where he worked on image and video compression topics. From 2013 to 2014, Dr. Akyol was a postdoctoral researcher in the Electrical Engineering Department at the University of Southern California and, between 2014 and 2017, in the Coordinated Science Laboratory at the University of Illinois at Urbana–Champaign. His current research focuses on information processing challenges associated with socio-cyber-physical systems. He is a Senior Member of IEEE and a 2021 NSF CAREER Award recipient.
\end{IEEEbiography}

\end{document}